\documentclass[11pt,USenglish]{article}

\usepackage[margin=1.0in]{geometry}
\usepackage[T1]{fontenc}
\usepackage[utf8]{inputenc}
\usepackage[english]{babel}

\usepackage{amsmath}
\usepackage{amsthm}
\usepackage{amssymb}
\usepackage{amsfonts}
\usepackage{xspace}
\usepackage[export]{adjustbox}

\usepackage{ifpdf}
\ifpdf
\usepackage[pdftex,colorlinks,allcolors=blue]{hyperref}
\else
\usepackage[hypertex,colorlinks,linkcolor=blue,citecolor=blue,filecolor=blue,urlcolor=blue]{hyperref}
\fi

\usepackage{mwe}
\usepackage{xcolor}
\usepackage{complexity}
\usepackage{graphicx}
\usepackage[export]{adjustbox}

\newtheorem{theorem}{Theorem}[section]
\newtheorem{lemma}[theorem]{Lemma}
\newtheorem{definition}[theorem]{Definition}

\theoremstyle{plain}
\newtheorem{claim}[theorem]{Claim}

\newcommand{\ProblemName}[1]{\textsf{#1}}
\newcommand{\MF}{\ProblemName{Max-Flow}\xspace}
\newcommand{\SSMF}{\ProblemName{Single-Source Max-Flow}\xspace}
\newcommand{\APMF}{\ProblemName{All-Pairs Max-Flow}\xspace}
\newcommand{\STMF}{\ProblemName{ST-Max-Flow}\xspace}
\newcommand{\GMF}{\ProblemName{Global Max-Flow}\xspace}
\newcommand{\MLEC}{\ProblemName{Maximum Local Edge Connectivity}\xspace}
\newcommand{\kPMF}{\ProblemName{kPMF}\xspace}
\newcommand\tO{\ensuremath{\tilde O}}
\newcommand{\Madry}{M{\k{a}}dry\xspace}
\providecommand{\card}[1]{\lvert#1\rvert}

\title{Conditional Lower Bounds for All-Pairs Max-Flow\footnote{This work was partially supported by the Israel Science Foundation grant \#897/13 and by a Minerva Foundation grant. An extended abstract of this article appears in Proceedings of ICALP 2017 and is also available at 
\href{https://arxiv.org/abs/1702.05805}{arXiv:1702.05805}. 
The most significant difference is the addition of Section~\ref{GTSC}.}}

\author{Robert Krauthgamer\footnote{Email: \texttt{robert.krauthgamer@weizmann.ac.il}  }
\qquad 
Ohad Trabelsi\footnote{Email: \texttt{ohad.trabelsi@weizmann.ac.il}}
\\
Weizmann Institute of Science
}

\begin{document}
\maketitle

\begin{abstract}
We provide evidence that computing the maximum flow value between 
every pair of nodes in a \emph{directed} graph on $n$ nodes, $m$ edges, 
and capacities in the range $[1..n]$, which we call the \APMF problem,
cannot be solved in time that is significantly faster
(i.e., by a polynomial factor) than $O(n^3)$ even for sparse graphs, namely $m=O(n)$; thus for general $m$, it cannot be solved significantly faster than $O(n^2m)$.
Since a single maximum $st$-flow can be solved 
in time $\tO(m\sqrt{n})$ [Lee and Sidford, FOCS 2014], 
we conclude that the all-pairs version might require time equivalent to $\tilde\Omega(n^{3/2})$ computations of maximum $st$-flow, 
which strongly separates the directed case from the undirected one.
Moreover, if maximum $st$-flow can be solved in time $\tO(m)$, 
then the runtime of $\tilde\Omega(n^2)$ computations is needed. This is in contrast to a conjecture of Lacki, Nussbaum, Sankowski, and Wulff-Nilsen [FOCS 2012] that \APMF in general graphs can be solved faster than the time of $O(n^2)$ computations of maximum $st$-flow.

Specifically, we show that in sparse graphs $G=(V,E,w)$, if one can compute the maximum $st$-flow from every $s$ in an input set of sources $S\subseteq V$ 
to every $t$ in an input set of sinks $T\subseteq V$ in time $O((\card{S}\card{T}m)^{1-\varepsilon})$, 
for some $\card{S}$, $\card{T}$ and a constant $\varepsilon>0$, 
then MAX-CNF-SAT (maximum satisfiability of conjunctive normal form formulas)
with $n'$ variables and $m'$ clauses can be solved in time ${m'}^{O(1)}2^{(1-\delta)n'}$ for a constant $\delta(\varepsilon)>0$, 
a problem for which not even $2^{n'}/\poly(n')$ algorithms are known. Such running time for MAX-CNF-SAT would in particular refute the Strong Exponential Time Hypothesis (SETH). Hence, we improve the lower bound of Abboud, Vassilevska-Williams, and Yu [STOC 2015], who showed that for every fixed $\varepsilon>0$ and $\card{S}=\card{T}=O(\sqrt{n})$, if the above problem can be solved in time $O(n^{3/2-\varepsilon})$, then some incomparable (and intuitively weaker) conjecture is false. Furthermore, a larger lower bound than ours implies strictly super-linear time
for maximum $st$-flow problem, which would be an amazing breakthrough.

In addition, we show that \APMF in \emph{uncapacitated} networks with every edge-density $m=m(n)$,
cannot be computed in time significantly faster than $O(mn)$, even for acyclic networks.
The gap to the fastest known algorithm by Cheung, Lau, and Leung [FOCS 2011]
is a factor of $O(m^{\omega-1}/n)$, and for acyclic networks it is $O(n^{\omega-1})$, 
where $\omega$ is the matrix multiplication exponent.

Finally, we extend our lower bounds to the version that asks only for 
the maximum-flow values below a given threshold
(over all source-sink pairs).
\end{abstract}


\section{Introduction}

The maximum flow problem is one of the most fundamental problems in combinatorial optimization. This classic problem and its variations such as minimum-cost flow, integral flow, and minimum-cost circulation, were studied extensively over the past decades, and have become key algorithmic tools
with numerous applications, in theory and in practice. 
Moreover, techniques developed for flow problems were generalized
or adapted to other problems, 
see for example \cite{bazaraa2011linear, AMJ93, arora2012multiplicative}.
The maximum $st$-flow problem, which we shall denote \MF, 
asks to ship the maximum amount of flow from a source node $s$ 
to a sink node $t$ in a directed edge-capacitated graph $G=(V,E,w)$,
where throughout, we denote $n=\card{V}$ and $m=\card{E}$, 
and assume integer capacities bounded by $U$.
After this problem was introduced in 1954 by Harris and Ross 
(see~\cite{schrijver2002history} for a historical account),
Ford and Fulkerson~\cite{FF56} devised the first algorithm for \MF, 
which runs in time $O((n+m)F)$, 
where $F$ is the maximum value of a feasible flow. 
Ever since, a long line of generalizations and improvements was studied, 
and the current fastest algorithm for \MF with arbitrary capacities is by Lee and Sidford~\cite{lee2014path}, 
which takes ${O}(m\sqrt{n}\log U)$ time. 
For the case of small capacities and sufficiently sparse graphs, the fastest algorithm, due to \Madry~\cite{madry2016computing}, 
has a running time  $\tO(m^{10/7}U^{1/7})$.
Here and throughout, $\tO(f)$ denotes $O(f\log^c f)$ for unspecified constant $c>0$.

A very natural problem is to compute the maximum $st$-flow for multiple 
source-sink pairs in the same graph $G$. 
The seminal work of Gomory and Hu~\cite{GH61} shows that in undirected graphs, \MF for all $\binom{n}{2}$ source-sink pairs requires at most $n-1$ 
executions of \MF (see also~\cite{Gusfield90}, where the $n-1$ computations are all on the input graph),
and a lot of research aimed to extend this result to directed graphs, 
with several partial successes, see details in Section~\ref{Prior}.
However, it is still not known how to solve \MF for multiple source-sink pairs
faster than solving it separately for each pair,
even in special cases like a single source and all possible sinks.
We shall consider the following problems involving multiple source-sink pairs, where the goal is always to report the value of each flow
(and not an actual flow attaining it).

\begin{definition}(\SSMF)
Given a directed edge-capacitated graph $G=(V,E,w)$ and a source node $s\in V$, output, for every $t\in V$, the maximum flow that can be shipped in $G$ from $s$ to $t$.
\end{definition}

\begin{definition}(\APMF)
Given a directed edge-capacitated graph $G=(V,E,w)$, output, for every pair of nodes $u,v\in V$, the maximum flow that can be shipped in $G$ from $u$ to $v$.
\end{definition}

\begin{definition}(\STMF)
Given a directed edge-capacitated graph $G=(V,E,w)$ and two subsets of nodes $S,T\subseteq V$, output, for every pair of nodes $s\in S$ and $t\in T$, the maximum flow that can be shipped in $G$ from $s$ to $t$.
\end{definition}

\begin{definition}(\GMF)
Given a directed edge capacitated graph $G=(V,E,w)$, output the maximum among all pairs $u,v\in V$, of the maximum flow value that can be shipped in $G$ from $u$ to $v$.
\end{definition}

\begin{definition}(\MLEC)
Given a directed graph $G=(V,E)$, output the maximum among all pairs $u,v\in V$, of the maximum number of edge-disjoint $uv$-paths in $G$.
\end{definition}

Note that in a graph with all edge capacities equal to $1$, the problem of finding the maximum local edge connectivity is equivalent to finding the global maximum flow.

\subsection{Prior Work}\label{Prior}

\begin{table}[t]
\centering
\begin{tabular}{c c c c c}
\hline\hline
Directed & Class & Problem & Runtime & Reference  \\ [0.5ex]
\hline
No & General & All-Pairs (G-H Tree) & $(n-1)T(n,m)$ & ~\cite{GH61}\\
No & Uncapacitated Networks & All-Pairs (G-H Tree) & $\tO(mn)$ & ~\cite{KM02},~\cite{BHKP07} \\
No & Genus bounded by $g$ & All-Pairs (G-H Tree) & $2^{O(g^2)}n\log^3n$ & ~\cite{borradaile2014all} \\
Yes & Sparse & All-Pairs & $O(n^2+\gamma^4\log \gamma)$ & ~\cite{ArikatiCZ95}\\
Yes & Constant Treewidth & All-Pairs & $O(n^2)$  & ~\cite{ArikatiCZ95} \\
Yes & Uncapacitated & All-Pairs & $O(m^{\omega})$  & ~\cite{cheung2013graph}\\
Yes & Uncapacitated DAG & Single-Source & $O(n^{\omega-1} m)$  & ~\cite{cheung2013graph}\\
Yes & Planar & Single-Source & $O(n\log^3 n)$  & ~\cite{lacki2012single}\\ [1ex]
\hline\hline
\end{tabular}
\caption{Known algorithms for multiple-pairs \MF. 
In this table, $T(n,m)$ is the fastest time to compute maximum $st$-flow 
in an undirected graph, $\omega$ is the matrix multiplication exponent, and $\gamma=\gamma(G)$ is a topological property of the input network that varies between $1$ and $\Theta(n)$. 
In planar graphs, $\gamma$ is the minimum number of faces required to cover all the nodes (i.e., every node is adjacent to at least one such face) 
over all possible planar embeddings~\cite{Frederickson95}.}
\vspace{.1in}\hrule
\label{table:nonlin}

\end{table}

We start with undirected graphs, where the \APMF values can be represented in a very succint manner, called nowdays \textit{a Gomory-Hu} tree~\cite{GH61}. In addition to being very succint, it allows the flow values and the corresponding cuts (vertex partitions) to be quickly retrieved. For a list of previous algorithms for multiple pairs maximum $st$-flow, see Table~\ref{table:nonlin}.
For directed graphs, no current algorithm computes the maximum flow between any $k=\omega(1)$ given pairs of nodes faster than the time of $O(k)$ separate \MF computations. However, some results are known in special settings. 
It is possible to compute \MF for $O(n)$ pairs in the time it takes for a single \MF computation~\cite{hao1994faster} and this result is used to find a global minimum cut. However, these pairs cannot be specified in the input.

For directed planar graphs, there is an $O(n\log^3n)$ time algorithm for the \SSMF problem~\cite{lacki2012single}, which immediately yields an $O(n^2\log^3n)$ time algorithm for the All-Pairs version, that is much faster than the time of $O(n^2)$ computations of planar \MF, a problem that can be solved in time $O(n\log n)$~\cite{borradaile2009n}. Based on these results, it was conjectured in~\cite{lacki2012single} that also in general graphs, \APMF can be solved faster than the time required for computing $O(n^2)$ separate maximum $st$-flows.
 
Several hardness results are known for multiple-pairs variants of \MF~\cite{AbboudWY15}. For \STMF in sparse graphs ($m=O(n)$) and $\card{S}=\card{T}=O(\sqrt{n})$, there is an $n^{3/2-o(1)}$ lower bound assuming at least one of the Strong Exponential Time Hypothesis (SETH), 3SUM, and All-Pairs Shortest-Paths (APSP) conjectures is correct (for comprehensive surveys on them, see~\cite{VWSurvey,Vsurvey18}). In addition, they show that \SSMF on sparse graphs requires $n^{2-o(1)}$ time, unless MAX-CNF-SAT can be solved in time $2^{(1-\delta)n}\poly(m)$ for some fixed $\delta>0$, and in particular SETH is false.

We will rely on SETH, a conjecture introduced by~\cite{ImpaSETH}, and on some weaker assumption related to its maximization version, MAX-CNF-SAT. In more detail, SETH states that for every fixed $\varepsilon>0$ there is an integer $k\geq 3$ such that kSAT on $n$ variables and $m$ clauses cannot be solved in time $2^{(1-\varepsilon)n}\poly(m)$, where $\poly(m)$ refers to $O(m^c)$ for unspecified constant $c$. By the sparsification lemma~\cite{Impa01spar}, in order to refute SETH it can be assumed that the number of clauses is $O(n)$. The MAX-CNF-SAT problem asks for the maximum number of clauses that can be satisfied in an input CNF formula. Most of our conditional lower bounds are based on the assumption that for every fixed $\delta>0$, MAX-CNF-SAT cannot be solved in time $2^{(1-\delta)n}\poly(m)$, 
where currently even $2^n/\poly(n)$ algorithms are not known for this problem~\cite{AbboudWY15}. Note that this is a weaker assumption than SETH, since a faster algorithm for MAX-CNF-SAT would imply a faster algorithm for CNF-SAT and refute SETH. Different assumptions regarding the hardness of CNF-SAT have been the basis for many lower bounds, including for the runtime of solving NP-hard problems exactly, parametrized complexity, and problems in P. See the Introduction in~\cite{abboud2017seth} and the references therein.

\subsection{Our Contribution}

We present conditional runtime lower bounds for both uncapacitated and capacitated networks. The proofs appear in sections~\ref{ProofsUncap} and~\ref{ProofsCap}, respectively, where the order reflects increasing level of complication. All our lower bounds hold even when the input $G$ is a DAG and has a constant diameter, and in the case of general capacities, they can be easily modified to apply also for graphs with constant maximum degree. In addition, for integer $k\geq 1$ we use $[k]$ to denote the range $\{1,...,k\}$.

\paragraph*{Capacitated Networks}

Our main result is that for every set sizes $\card{S}$ and $\card{T}$, the \STMF cannot be solved significantly faster than $O(\card{S}\card{T}m)$ (i.e., polynomially smaller runtime), unless a breakthrough in MAX-CNF-SAT is achieved, and consequently in SETH.

\begin{theorem}\label{ST}
If for some fixed constants $\varepsilon>0$, $c_1,c_2\in [0,1]$, \STMF on graphs with $n$ nodes, $\card{S}=\tilde{\Theta}(n^{c_1})$, $\card{T}=\tilde{\Theta}(n^{c_2})$, $m=O(n)$ edges, and capacities in $[n]$ can be solved in time $O((\card{S}\card{T}m)^{1-\varepsilon})$, then for some $\delta(\varepsilon)>0$, MAX-CNF-SAT on $n'$ variables and $O(n')$ clauses can be solved in time $O(2^{(1-\delta)n'})$, and in particular SETH is false.
\end{theorem}
 
This result improves the aforementioned $n^{3/2-o(1)}$ lower bound of~\cite{AbboudWY15}, as for their setting of $\card{S}=\card{T}=O(\sqrt{n})$ our lower bound is $n^{2-o(1)}$, although their lower bound is based on an incomparable (and intuitively weaker) conjecture, that at least one of the SETH, 3SUM, and APSP conjectures is correct. 
In fact, if there was a reduction from SETH that implied a larger runtime lower bound for \STMF, then the (single-pair) \MF problem would require a strictly super-linear time under it, 
but such a reduction is not possible unless the non-deterministic version of SETH (abbreviated NSETH) is false~\cite{carmosino2016nondeterministic}. 
And anyway, such a lower bound for \MF would be an amazing breakthrough.

The next theorem is an immediate corollary of Theorem~\ref{ST}, by assigning $\card{S},\card{T}=\Theta(n)$.

\begin{theorem}\label{AllPairs}
If for some fixed $\varepsilon>0$, \APMF in graphs with $n$ nodes, $m=O(n)$ edges, and capacities in $[n]$ can be solved in time $O((n^2m)^{1-\varepsilon})$, then for some $\delta(\varepsilon)>0$, MAX-CNF-SAT on $n'$ variables and $O(n')$ clauses can be solved in time $O(2^{(1-\delta)n'})$, and in particular SETH is false.
\end{theorem}

This conditional lower bound (see Figure~\ref{figs:SOTA_APMF_2}) shows that \APMF requires time that is equivalent to $\Omega(n^{3/2})$ computations of \MF, which strongly separates the directed case from the undirected one (where a Gomory-Hu tree can be constructed in the time of $n-1$ computations). If \MF takes $\tO(m)$ time, which is currently open but plausible, then the running time of $\tilde{\Omega}(n^2)$ computations of \MF is needed. This is in contrast to the aforementioned conjecture of Lacki, Nussbaum, Sankowski, and Wulf-Nilsen~\cite{lacki2012single} that \APMF in general graphs can be solved faster than the time of $O(n^2)$ computations of maximum $st$-flow.

\paragraph*{Uncapacitated Networks}

For the case of uncapacitated networks, we show that for every $m=m(n)$, \APMF cannot be solved significantly faster than $O(mn)$.
Here we introduce a new technique to design reductions from SETH to graphs with varying edge densities, rather than the usual reductions that only deal with sparse graphs. 

\begin{theorem}\label{UncapAllPairs}
If for some fixed $\varepsilon>0$ and $c\in [1,2]$, \APMF in uncapacitated graphs with $n$ nodes and $m=\tilde{\Theta}(n^{c})$ edges can be solved in time $O((nm)^{1-\varepsilon})$, then for some $\delta(\varepsilon)>0$, MAX-CNF-SAT on $n'$ variables and $O(n')$ clauses can be solved in time $O(2^{(1-\delta)n'})$, and in particular SETH is false.
\end{theorem}

Hence, a certain additional improvement to the $O(m^{\omega})$ time algorithm of~\cite{cheung2013graph} (and similarly to the $O(n^{\omega}m)$ time for DAGs, where our lower bounds apply too) is not likely.
We now present conditional lower bounds for \STMF, which are functions of $\card{S}$ and $\card{T}$.
\begin{theorem}\label{UncapST}
If for some fixed constants $\varepsilon>0$, $c_1,c_2\in [0,1]$, \STMF on uncapacitated graphs with $n$ nodes, $\card{S}=\tilde{\Theta}(n^{c_1})$, $\card{T}=\tilde{\Theta}(n^{c_2})$, and $O((\card{S}+\card{T})n)$ edges can be solved in time $O((\card{S}\card{T}n)^{1-\varepsilon})$, then for some $\delta(\varepsilon)>0$, MAX-CNF-SAT on $n'$ variables and $O(n')$ clauses can be solved in time $O(2^{(1-\delta)n'})$, and in particular SETH is false.
\end{theorem}

\begin{figure}[!ht]
	\centering
		\includegraphics[width=1.0\textwidth,left]{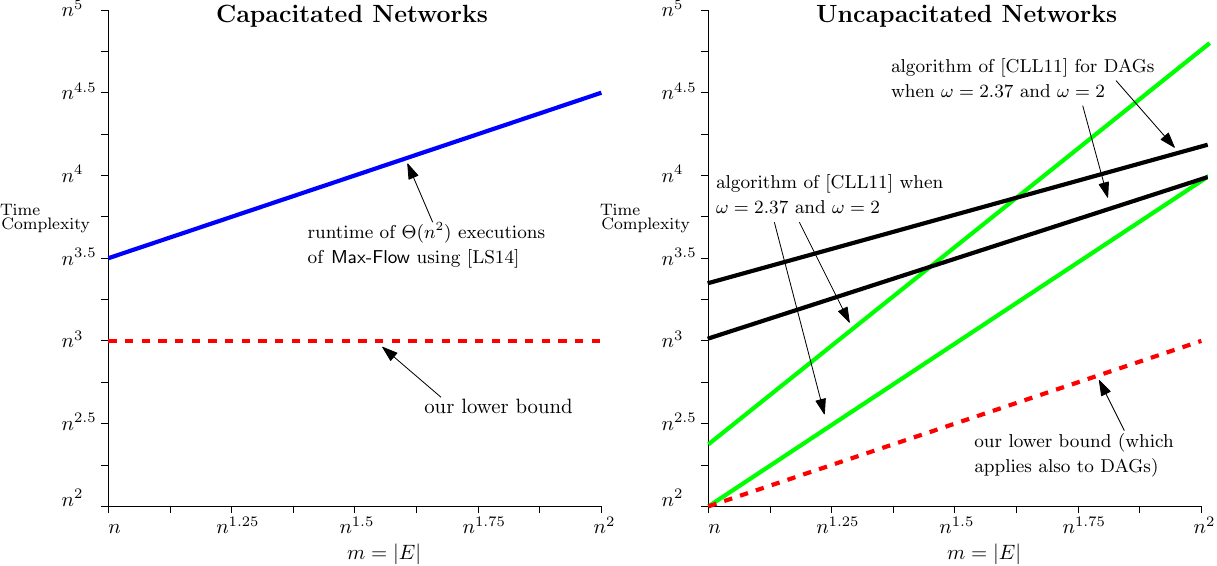}
   \caption[-]{State of the art bounds for \APMF in directed networks. Conditional lower bounds are depicted in dashed lines, and known algorithms in solid lines.
   }
   \label{figs:SOTA_APMF_2}
\vspace{.1in}\hrule
\end{figure}

In addition, we present a conditional lower bound for computing the \MLEC of sparse graphs, which is the same as \GMF if all the capacities are $1$, that is indeed the case in our reduction. The next result, proved in Section~\ref{GMF}, was obtained together with Bundit Laekhanukit and Rajesh Chitnis, and we thank them for their permission to include it here.

\begin{theorem}\label{MLEC}
If for some fixed $\varepsilon>0$, the \MLEC in graphs with $n$ nodes and $\tO(n)$ edges can be found in time $O(n^{2-\varepsilon})$, then for some $\delta(\varepsilon)>0$, MAX-CNF-SAT on $n'$ variables and $O(n')$ clauses can be solved in time $O(2^{(1-\delta)n'})$, and in particular SETH is false.
\end{theorem}

\paragraph*{Generalization to Bounded Cuts}
Finally, we show in Section~\ref{GTSC} that our lower bounds extend to 
the version that requires to output the maximum-flow value only for 
source-sink pairs for which this value is at most some given threshold $k$.

\paragraph*{Connection to the Orthogonal Vectors Problem}
Our techniques are based on partitioning the variable set of CNF-SAT to sets of different sizes, and constructing graphs with the property that certain pairs of nodes would have smaller maximum flow between them if and only if they correspond to a satisfying assignment. This approach is inspired by results of Williams~\cite{WilliamsOV}.

We remark that all of our theorems can also be proved assuming that for the appropriate $k\in\{2,3\}$, the \ProblemName{$k$-Orthogonal Vectors} ($k$OV) problem cannot be solved in time $\tO(n^{k-\varepsilon})$ for a fixed constant $\varepsilon>0$, in what is called the $k$OV Hypothesis (see~\cite{VWSurvey,Vsurvey18}). In the $k$OV problem the input is $k$ sets $\{U_i\}_{i\in [k]}$, each of $n$ vectors from $\{0,1\}^d$, and the goal is to find $k$ vectors $\{u_i\}_{i\in [k]}$, one from each set, such that $u_1\cdot...\cdot u_k:=\sum_{i=1}^{d} \prod_{j=1}^{k} u_j[i]= 0$ (for $k=2$ it means that $u_1,u_2$ are orthogonal). An equivalent version of the problem has $U_1=...=U_k$. Solving $k$OV in time $O(n^kd)$ can be done easily by exhaustive search, while the fastest known algorithm for the problem runs in time $n^{k-1/\Theta(\log(d/\log n))}$~\cite{Abboud15,Chan16}. Williams~\cite{WilliamsOV} proved that SETH implies the non-existence of an $\tO(n^{k-\varepsilon})$-time algorithm.


\section{Reduction to Multiple-Pairs \ProblemName{Max-Flow} with Unit Capacity}\label{ProofsUncap}

In this section we prove Theorems~\ref{UncapAllPairs} and~\ref{UncapST}. We start with a general lemma which is the heart of the proofs.

\begin{lemma}\label{ProofsUncapLemma}

Let $a\in [0,1]$ and $b\in [0,1-a]$. Then MAX-CNF-SAT on $n$ variables and $m$ clauses $\{C_i\}_{i\in [m]}$ can be reduced to $O(m)$ instances of \STMF with $\card{S}=2^{an}$ and $\card{T}=2^{bn}$ in graphs with $\Theta(2^{an}+2^{(1-a-b)n}p+2^{bn})$ nodes, $\Theta((2^{an}+2^{bn})\cdot 2^{(1-a-b)n}m)$ edges, and capacities in $\{0,1\}$.
\end{lemma}

\begin{proof}
Given a CNF-formula $F$ on $n$ variables and $m$ clauses as input for MAX-CNF-SAT, $a\in [0,1]$, and $b\in [0,1-a]$, we split the variables into three sets $U_1$, $U_2$, and $U_3$, where $U_1$ is of size $a n$, $U_2$ is of size $(1-a-b) n$, and $U_3$ is of size $b n$, and enumerate all their $2^{an}$, $2^{(1-a-b)n}$, and $2^{bn}$ partial assignments (with respect to $F$), respectively, when the objective is to find a triple $(\alpha,\beta,\gamma)$ of assignments to $U_1$, $U_2$, and $U_3$ respectively, that satisfies the maximal number of clauses. We will have an instance $G_p$ of \STMF for each value $p\in [m]$, in which by one call to \STMF we check if there exists a triple $\alpha$, $\beta$, and $\gamma$ that satisfies at least $p$ clauses, as follows.

We construct a graph $G_p$ for every $p\in [m]$ on $N$ nodes $V_1\cup V_2\cup V_3$, where $V_1$ contains a node $\alpha$ for every assignment $\alpha$ to $U_1$, $V_2$ contains $2m+1+(p-1)=2m+p$ nodes for every assignment $\beta$ to $U_2$, that are $\beta_i^{l}$ and $\beta_i^{r}$ for every $i\in[m]$, $\beta'$, and the set $\{\beta'_i\}_{i\in [p-1]}$, and $V_3$ contains a node $\gamma$ for every assignment $\gamma$ to $U_3$. We use the notation $\alpha$ for nodes in $V_1$ and for assignments to $U_1$, $\beta$ for assignments to $U_2$, and $\gamma$ for nodes in $V_3$ and assignments to $U_3$. However, it will be clear from the context. Now, we have to describe the edges in the network. In order to simplify the reduction, we partition the edges into blue and red colors, as follows.

For every $\alpha$, $\beta$, and $i\in [m]$, we add a blue edge from $\alpha$ to $\beta_i^{l}$ if both of $\alpha$ and $\beta$ do not satisfy the clause $C_i$ (do not set any of the literals to true), and otherwise we add a red edge from $\alpha$ to $\beta_i^{r}$. We further add, for every $\beta$, $\gamma$, and $i\in [m]$, a blue edge from $\beta_i^{l}$ to $\gamma$ if $\gamma$ does not satisfy $C_i$. For every $\beta$, $\gamma$, and $j\in [p-1]$, we add a red edge from every $\beta'_j$ to every $\gamma$. For every $\beta$ and $i\in [m]$, we add a red edge from $\beta_i^{l}$ to $\beta_i^{r}$ and from $\beta_i^{r}$ to $\beta'$, and finally for every $\beta$ and $j\in [p-1]$, we add a red edge from $\beta'$ to $\beta'_j$, where all edges are of capacity $1$.

The graph we built has $2^{an}+2\cdot 2^{(1-a-b)n}m+2^{(1-a-b)n}+2^{(1-a-b)n}(p-1)+2^{bn}=\Theta(2^{an}+2^{(1-a-b)n}m+2^{bn})$ nodes, $2^{an}\cdot 2^{(1-a-b)n}m+2^{bn}\cdot 2^{(1-a-b)n}m+2\cdot 2^{(1-a-b)n}m+$ $(p-1)2^{(1-a-b)n}+2^{bn}\cdot (p-1)2^{(1-a-b)n}=\Theta((2^{an}+2^{bn})\cdot 2^{(1-a-b)n}m)$ edges, with capacities in $\{0,1\}$ (see Figure~\ref{Figs:REDUCTION_UNCAP}), and its construction time is asymptotically the same as the time it takes to output its edge set.

\begin{figure}[!ht]
       \includegraphics[width=0.9\textwidth,left]{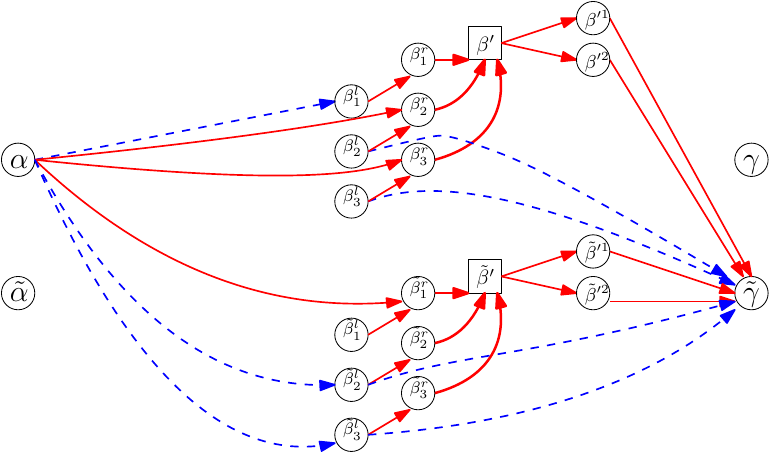}
   \caption[-]{An illustration of part of the reduction. Here, $U_1$, $U_2$, and $U_3$ have $2$ assignments each, $\alpha$ and $\tilde{\alpha}$ to $U_1$, $\beta$ and $\tilde{\beta}$ to $U_2$, and $\gamma$ and $\tilde{\gamma}$ to $U_3$. Blue edges are dashed. For simplicity, only the edges of $G_3^{\alpha,\beta,\tilde{\gamma}}\cup G_3^{\alpha,\tilde{\beta},\tilde{\gamma}}$ are presented. In this illustration, $\alpha$ does not satisfy anything, $\beta$ satisfies $C_2$ and $C_3$, $\tilde{\beta}$ satisfies $C_1$, and $\tilde{\gamma}$ satisfies $C_1$. Note that the assignment comprised of $\alpha$, $\beta$, and $\tilde{\gamma}$ satisfies all the clauses, and indeed the maximum flow from $\alpha$ to $\gamma$ is $2\cdot 3 -1=5$.
   }
   \label{Figs:REDUCTION_UNCAP}
\vspace{.1in}\hrule
\end{figure}

For every $\alpha$, $\beta$, and $\gamma$, we denote by $G_p^{\alpha,\beta,\gamma}$ the graph induced from $G_p$ on the nodes

$$\{\alpha, \beta', \gamma \} \cup \Bigg( \bigcup_{\substack{y\in \{l, r\} \\ i\in [m]}} \{\beta_i^{y}\}\Bigg) \cup \Bigg( \bigcup_{\substack{j\in [p-1]}} \{\beta'_j\}\Bigg).$$

We claim that for every $\alpha$ and $\gamma$, the maximum flow from $\alpha$ to $\gamma$ can be bounded by the sum, over all $\beta$, of the maximum flow between them in $G_p^{\alpha,\beta,\gamma}$. This claim follow easily because the intersection $G_p^{\alpha,\beta_1,\gamma}\cap G_p^{\alpha,\beta_2,\gamma}$ for $\beta_1\neq\beta_2$ is exactly the source and the sink $\{\alpha,\gamma\}$, no edge passes between these two graphs, and $\Big( \bigcup_{\substack{\beta}} G_i^{\alpha,\beta,\gamma}\Big)$ consists of all nodes that are both reachable from $\alpha$ and $\gamma$ is reachable from them.

We now prove that if there is an assignment to $F$ that satisfies at least $p$ clauses then the graph $G_p$ we built has a triple $\alpha, \beta, \gamma $ with maximum flow from $\alpha$ to $\gamma$ in $G_p^{\alpha,\beta,\gamma}$ at most $m-1$. Since for every $\tilde{\beta}$, $m$ is the number of outgoing edges from $\alpha$ in $G_p^{\alpha,\tilde{\beta},\gamma}$, $m$ is also an upper bound for the maximum flow from $\alpha$ to $\gamma$ in it, and hence in $G_p$ it is at most $2^{(1-a-b)n}m-1$. Otherwise, we will show that every triple $\alpha, \beta, \gamma$ has a maximum flow from $\alpha$ to $\gamma$ in $G_p^{\alpha,\beta,\gamma}$ of size at least $m$, and so in $G_p$ it is at least $2^{(1-a-b)n}m$. Hence, by simply picking the maximal $j\in [m]$ such that the maximum flow in $G_j$ of some pair $\alpha,\gamma$ is at most $2^{(1-a-b)n}m-1$, and then by iterating over all assignments $\beta$ to $U_2$ with $\alpha$ and $\gamma$ fixed as the assignments to $U_1$ and $U_3$, we can also find the required triple $\alpha,\beta,\gamma$.

For the first direction, assume that $F$ has an assignment that satisfies at least $p$ clauses, and denote such assignment by $\Phi$. Let $\alpha_{\Phi}$, $\beta_{\Phi}$, and $\gamma_{\Phi}$ be the assignments to $U_1$, $U_2$, and $U_3$, respectively, that are induced from $\Phi$. Since a blue path from $\alpha_{\Phi}$ through ${(\beta_{\Phi})}_i^l$ for some $i\in [m]$ to $\gamma_{\Phi}$ corresponds to $\alpha_{\Phi}$, $\beta_{\Phi}$, and $\gamma_{\Phi}$ all do not satisfy $C_i$, in $G_p^{\alpha_{\Phi},\beta_{\Phi},\gamma_{\Phi}}$ there are at most $m-p$ (internally) disjoint blue paths from $\alpha$ to $\gamma$. As the only way to ship flow in $G_p^{\alpha_{\Phi},\beta_{\Phi},\gamma_{\Phi}}$ that is not through a blue path is through the node $\beta_{\Phi}'$, and the total number of edges going out of this node is $p-1$, we conclude that the total maximum flow in $G_p^{\alpha_{\Phi},\beta_{\Phi},\gamma_{\Phi}}$ from $\alpha_{\Phi}$ to $\beta_{\Phi}$ is bounded by $m-p+(p-1) = m-1$. Since for every $\beta$, the maximum amount of flow that can be shipped in $G_p^{\alpha_{\Phi},\beta,\gamma_{\Phi}}$ from $\alpha_{\Phi}$ to $\gamma_{\Phi}$ is at most $m$, summing over all $\beta$ we get that the total flow in $G_p$ from $\alpha_{\Phi}$ to $\gamma_{\Phi}$ is bounded by $(2^{(1-a-b)n}-1)m+(m-1)\leq 2^{(1-a-b)n}m-1$, as required.

For the second direction, assume that every assignment to $F$ satisfies at most $p-1$ clauses. In order to show that the maximum flow from every $\alpha$ to every $\gamma$ is at least $2^{(1-a-b)n}m$, we first fix $\alpha$, $\beta$, and $\gamma$. Then, by passing flow in two phases we show that $m$ units of flow can be passed in $G_p^{\alpha,\beta,\gamma}$ from $\alpha$ to $\gamma$. As this argument applies for every $\beta$, we can add up the respective flows without violating capacities, concluding the proof. By the assumption, there exist $m-(p-1) = m-p+1$ $i$'s, such that $\alpha$, $\beta$, and $\gamma$ do not satisfy $C_i$, and we denote a set with this amount of such $i$'s by $I_{\beta}$. Each of these $i$'s induces a blue path $(\alpha\rightarrow \beta_i^l\rightarrow \gamma)$ from $\alpha$ to $\gamma$ in $G_p^{\alpha,\beta,\gamma}$, and so we ship a unit of flow through every one of them according to $I_{\beta}$, in what we call the first phase. In the second phase, we ship additional $m-(m-p+1)=p-1$ units in the following way. Let $A_1:=\{i\in [m]\setminus I_{\beta} : \alpha\nvDash C_i \wedge \beta\nvDash C_i\}$, and $A_2:=([m]\setminus I_{\beta})\setminus A_1 = \{i\in [m]\setminus I_{\beta} : \alpha\vDash C_i \vee \beta\vDash C_i\}$, where $\alpha\vDash C_i$ denotes that the assignment $\alpha$ satisfies $C_i$ (as defined earlier), and $\alpha\nvDash C_i$ denotes that it does not satisfy $C_i$. Let $f:A_1\cup A_2\rightarrow [m-\card{I_{\beta}}]$ be a bijective function such that the range of $A_1$ is $[\card{A_1}]$ and the range of $A_2$ is $[m-\card{I_{\beta}}]\setminus [\card{A_1}]$. Clearly, there exists such bijection and it is easy to find one. For every $i\in A_1$ we ship flow through the path $(\alpha\rightarrow\beta_i^{l}\rightarrow \beta_i^{r}\rightarrow \beta'\rightarrow\beta'_j\rightarrow \gamma)$, and for every $i\in A_2$ through the path $(\alpha\rightarrow \beta_i^{r}\rightarrow \beta'\rightarrow\beta'_j\rightarrow \gamma)$, in both cases with $j=f(i)$. 

Since we defined the flow in paths, we only need to show that the capacity requirements hold, and we start with blue edges. Indeed, edges of the form $(\alpha, \beta_i^l)$ are used in the first phase, with flow that is determined uniquely by $\beta$ and $i\in I_{\beta}$, and in the second phase uniquely according to $\beta$ and $i\in [m]\setminus I_{\beta}$, and so they cannot be used twice. Edges of the form $(\beta_i^l,\gamma)$ are only used in the first phase, and their flow is uniquely determined according to $\beta$ and $i\in I_{\beta}$, and so are good too. We now proceed to red edges, which were used only in the second phase. 

Edges of the forms $(\alpha,\beta_i^r)$, $(\beta_i^l,\beta_i^r)$ and $(\beta_i^r,\beta')$ have flow that is uniquely determined by $\beta$ and $i\in [m]\setminus I_{\beta}$, and so are not used more than once. Edges of the form $(\beta',\beta'_j)$ have flow that is uniquely determined by $\beta$ and $j=f(i)\in [p-1]$, and since $f$ is a bijection, every $j$ has at most one $i$ such that $f(i)=j$, and so these edges are also used at most once. As a byproduct, and since every edge of the form $(\beta'_j,\gamma)$ has only the edge $(\beta',\beta'_j)$ as its source for flow, edges of the form $(\beta'_j,\gamma)$ are also used at most once. Altogether, we have bounded the total flow in all edges that were used in both phases, and so the capacity requirements follow, which completes the proof of the second direction and of Lemma~\ref{ProofsUncapLemma}.
\end{proof}

\begin{proof}[Proof of Theorem~\ref{UncapAllPairs}]
We apply Lemma~\ref{ProofsUncapLemma} in as follows. For every setting of $a=b\in [1/3,1/2]$ we get graphs $G=(V,E,w)$ with $\card{V}=\Theta(2^{an})$ ($\card{V}=\Theta(2^{an})m$ if $a=1/3$) and $\card{E}=\Theta(2^{(1-a)n}m)$. Hence, $\card{E}=\tilde{\Theta}(\card{V}^{1/a-1})$ and so in order to get any $c\in[1,2]$ we can pick $a(=b)$ such that additionally $c=1/a-1$, and Theorem~\ref{UncapAllPairs} follows.
\end{proof}

\begin{proof}[Proof of Theorem~\ref{UncapST}]
Here we apply Lemma~\ref{ProofsUncapLemma} a bit differently. For every setting of $a,b\in [0,1/2]$ such that $1-a-b\geq\max(a,b)$ we get graphs $G=(V,E,w)$ with $\card{V}=\Theta(2^{(1-a-b)n}m)$ and $\card{E}=\Theta((2^{an}+2^{bn})2^{(1-a-b)n}m)$. Hence, in order to get any $c_1, c_2\in [0,1]$, we can pick $a=c_1/(1+c_1+c_2)$ and $b=c_2/(1+c_1+c_2)$, which clearly satisfy the required conditions. Now, observe that $c_1=a/(1-a-b)$ and $c_2=b/(1-a-b)$, thus $\card{S}=(\card{V}/m)^{c_1}$ and $\card{T}=(\card{V}/m)^{c_2}$, we get our lower bound for $\card{E}=O((\card{S}+\card{T})\card{V})$, and Theorem~\ref{UncapST} follows.
\end{proof}

\section{Reduction to Multiple-Pairs \ProblemName{Max-Flow} in Capacitated Networks}\label{ProofsCap}
In this section we prove Theorems~\ref{ST} and~\ref{AllPairs}. We proceed to prove our main technical lemma.


\begin{lemma}\label{ProofsCapLemma}
Let $a\in [0,1]$ and $b\in [0,1-a]$. Then MAX-CNF-SAT on $n$ variables and $m$ clauses $\{C_i\}_{i\in [m]}$ can be reduced to $O(m)$ instances of \STMF with $\card{S}=2^{an}$ and $\card{T}=2^{bn}$ in graphs with $N=\Theta(2^{an}+2^{(1-a-b)n}m+2^{bn})$ nodes, $O((2^{an}+2^{(1-a-b)n}+2^{bn})m)=O(N)$ edges, and with capacities in $[N]$.
\end{lemma}

\begin{proof}
Given a CNF-formula $F$ on $n$ variables and $m$ clauses as input for MAX-CNF-SAT, $a\in [0,1]$, and $b\in [0,1-a]$, we begin similarly to before by splitting the variables into three sets $U_1$, $U_2$, and $U_3$ where $U_1$ is of size $an$, $U_2$ is of size $(1-a-b)n$, and $U_3$ is of size $bn$, and enumerate all their $2^{an}$, $2^{(1-a-b)n}$, and $2^{bn}$ partial assignments (with respect to $F$), respectively, when the objective is to find a triple $(\alpha,\beta,\gamma)$ of assignments to $U_1$, $U_2$, and $U_3$, that satisfy the maximal number of clauses. We will have an instance $G_p$ of \STMF for each value $p\in [m]$, in which by one call to \STMF we check if there exists a triple $(\alpha,\beta,\gamma)$ that satisfies at least $p$ clauses, as follows.

We construct the graph $G_p$ on $N$ nodes $V_1\cup V_2\cup V_3\cup A\cup B\cup \{v_B\}$, where $V_1$ contains a node $\alpha$ for every assignment $\alpha$ to $U_1$, $V_2$ contains $3m+1$ nodes for every assignment $\beta$ to $U_2$, that are $\beta_i^l$, $\beta_i^c$, $\beta_i^r$, for every $i\in[m]$, and $\beta'$, $V_3$ contains a node $\gamma$ for every assignment $\gamma$ to $U_3$, $A$ contains two nodes $C_i^{\vDash}$ and $C_i^{\nvDash}$ for every clause $C_i$, and $B$ contains a node $C_i$ for every clause $C_i$. We use the notation $\alpha$ for nodes in $V_1$ and assignments to $U_1$, $\beta$ to assignments to $U_2$, $\gamma$ for nodes in $V_3$ and assignments to $U_3$, and $C_i$ for nodes in $B$ and clauses. However, it will be clear from the context. Now, we have to describe the edges in the network. In order to simplify the reduction, we partition the edges into red and blue colors, as follows.

For every $\alpha$ and $i\in [m]$ we add a red edge of capacity $2^{(1-a-b)n}$ from $\alpha$ to $C_i^{\vDash}$ if $\alpha\vDash C_i$, and a blue edge of the same capacity from $\alpha$ to $C_i^{\nvDash}$ otherwise. We further add, for every $\beta$, a red edge of capacity $1$ from $C_i^{\vDash}$ to $\beta_i^c$, a blue edge of capacity $1$ from $C_i^{\nvDash}$ to $\beta_i^l$, a blue edge of capacity $1$ from $\beta_i^l$ to $\beta_i^r$ if $\beta\nvDash C_i$, a red edge of capacity $1$ from $\beta_i^c$ to $\beta'$, and a blue edge of capacity $1$ from $\beta_i^r$ to $C_i$. For every $\beta$ we add a red edge of capacity $p-1$ from $\beta'$ to $v_B$.
For every $\gamma$ we add a red edge of capacity $2^{(1-a-b)n}(p-1)$ from $v_B$ to $\gamma\in V_3$, and finally, for every $\gamma$ and $i\in [m]$ we add a blue edge of capacity $2^{(1-a-b)n}$ from $C_i$ to $\gamma$ if $\gamma\nvDash C_i$.

The graph we built has $N=2^{an}+2m+2^{(1-a-b)n}\cdot 3m+2^{(1-a-b)n}+1+m+2^{bn}=\Theta(2^{an}+2^{(1-a-b)n}\cdot m + 2^{bn})$ nodes, at most $2^{an}m+ 2^{(1-a-b)n}\cdot 2m + 2^{(1-a-b)n}\cdot 2m+ 2^{(1-a-b)n}m + 2^{(1-a-b)n}+1+ 2^{(1-a-b)n}m + 2^{bn}m=O((2^{an}+2^{(1-a-b)n}+2^{bn})m)$ edges, all of its capacities are in $[N]$, and its construction time is $O(Nm)$ (see Figure~\ref{Figs:REDUCTION_CAP}). 

\begin{figure}[!ht]
       \includegraphics[width=0.9\textwidth,left]{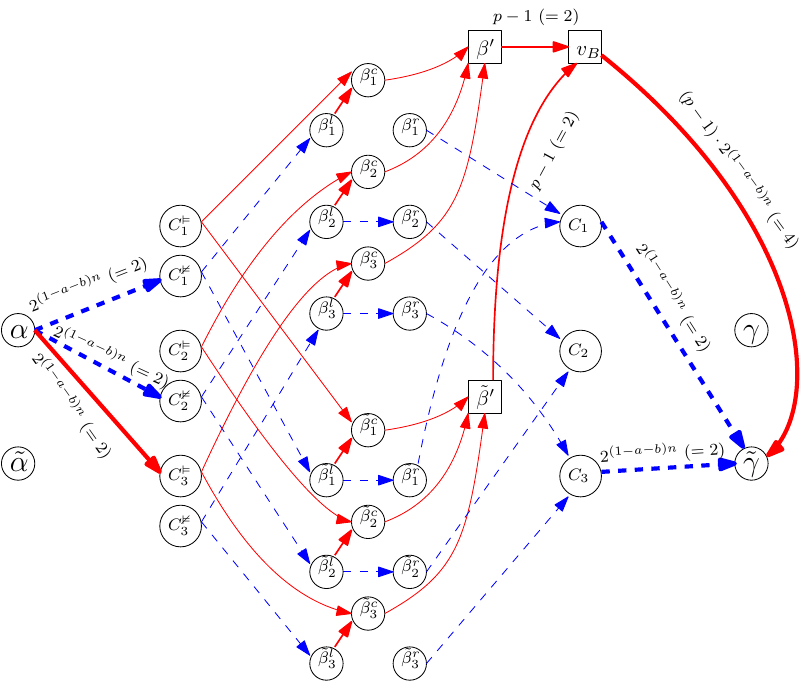}
   \caption[-]{An illustration of part of the reduction, with $p=m$. Here, $U_1$, $U_2$, and $U_3$ have $2$ assignments each; $\alpha$ and $\tilde{\alpha}$ to $U_1$, $\beta$ and $\tilde{\beta}$ to $U_2$, $\gamma$ and $\tilde{\gamma}$ to $U_3$. Bolder edges correspond to edges of higher capacity (specified wherever they are bigger than $1$), and blue edges are dashed. For simplicity, only the edges relevant to $\alpha$ and $\tilde{\gamma}$ are presented. In this illustration, $\alpha$ satisfies $C_3$, $\beta$ satisfies $C_1$, $\tilde{\beta}$ satisfies $C_3$, and $\tilde{\gamma}$ satisfies $C_2$. Note that the assignment comprised of $\alpha$, $\beta$, and $\tilde{\gamma}$ satisfies all the clauses, and indeed the maximum flow from $\alpha$ to $\gamma$ is $2\cdot 3 -1=5$.
   }
   \label{Figs:REDUCTION_CAP}
\vspace{.1in}\hrule
\end{figure}

We proceed to prove that if there is an assignment to $F$ that satisfies at least $p$ clauses then the graph $G_p$ we built has a pair $\alpha, \gamma$ with maximum flow from $\alpha$ to $\gamma$ at most $2^{(1-a-b)n}m-1$, and otherwise, every $\alpha, \gamma$ has a maximum flow of size at least $2^{(1-a-b)n}m$. Hence, by simply picking the maximal $j\in [m]$ such that the maximum flow in $G_j$ of some pair $\alpha,\gamma$ is at most $2^{(1-a-b)n}m-1$, and then by iterating over all assignments $\beta$ to $U_2$ with $\alpha$ and $\gamma$ fixed as the assignments to $U_1$ and $U_3$, we can also find the required triple $\alpha,\beta,\gamma$.

For the first direction, assume that $F$ has an assignment that satisfies at least $p$ clauses, and denote such assignment by $\Phi$. Let $\alpha_{\Phi}$, $\beta_{\Phi}$, and $\gamma_{\Phi}$ be the assignments to $U_1$, $U_2$, and $U_3$, respectively, that are induced from $\Phi$. We will show that there exists an $(\alpha_{\Phi}, \gamma_{\Phi})$ cut whose capacity is at most $2^{(1-a-b)n}m-1$, hence by the Min-Cut Max-Flow theorem, the maximum flow from $\alpha_{\Phi}$ to $\gamma_{\Phi}$ is bounded by this number, concluding the proof of the first direction. We define the cut in a way that for every $\beta\neq \beta_{\Phi}$, the cut will have $m$ cut edges that are contributed from nodes related to $\beta$, and nodes related to $\beta_{\Phi}$ will be carefully added to either side of the cut, so that they will contribute capacity of only $m-1$ to the cut. This is done by exploiting the fact that there are at most $m-p$ blue paths from $\alpha_{\Phi}$ to $\gamma_{\Phi}$ through nodes associated with $\beta_{\Phi}$. To be more precise, we define a suitable cut as follows. \footnote{In previous versions including the journal version (TALG 2018), $V_3\setminus \{\gamma_{\Phi}\}$ was erroneously not included in the set $S$.} 

$$
S=\{\alpha_{\Phi}, \beta_{\Phi}'\}\cup
\{C_i^{\vDash} : \alpha_{\Phi}\vDash C_i \}\cup 
\{C_i^{\nvDash} : \alpha_{\Phi}\nvDash C_i \}\cup
\{({\beta_{\Phi})}_i^c : i\in [m] \}\cup
\{C_i, {(\beta_{\Phi})}_i^l, {(\beta_{\Phi})}_i^r : \gamma_{\Phi}\vDash C_i \}\cup
$$
$$
\{{(\beta_{\Phi})}_i^l : \gamma_{\Phi}\nvDash C_i \wedge \beta_{\Phi}\vDash C_i\}\cup
V_3\setminus \{\gamma_{\Phi}\}
$$

\begin{claim}\label{CutLemma}
The cut $(S,V\setminus S)=(S,T)$ has capacity at most $2^{(1-a-b)n}m-1$.
\end{claim}

\begin{proof}[Proof of Claim]
We will go over all the nodes in $S$, and count the total capacity leaving to nodes in $T$ for each of them.
$\alpha_{\Phi}\in S$ and all nodes $C_i^{\vDash}$ and $C_i^{\nvDash}$ that are adjacent to it are in $S$ too, hence it does not contribute anything. For every $i\in [m]$, we have two cases for nodes in $A$. If $\alpha_{\Phi}\vDash C_i$ then $C_i^{\nvDash}\in T$ and hence $C_i^{\nvDash}$ does not contribute anything.
However, $C_i^{\vDash}$ has $2^{(1-a-b)n}$ outgoing edges, where all except ${(\beta_{\Phi})}_i^c$ are in $T$. Hence, it contributes $2^{(1-a-b)n}-1$ to the cut.
Else, if $\alpha_{\Phi}\nvDash C_i$ then $C_i^{\vDash}\in T$ and hence $C_i^{\vDash}$ does not contribute anything.
But $C_i^{\nvDash}$ has $2^{(1-a-b)n}$ outgoing edges, of which $2^{(1-a-b)n}-1$ are cut edges as their targets are in $T$, and the one incoming to ${(\beta_{\Phi})}_i^l$ is a cut edge if and only if ${(\beta_{\Phi})}_i\nvDash C_i$ and also $\gamma_{\Phi}\nvDash C_i$ (equivalently, ${(\beta_{\Phi})}_i^l\in T$), and in our current case it means that $\Phi\nvDash C_i$.
Hence, for every $i\in [m]$, the nodes in $\{C_i^{\vDash},C_i^{\nvDash}\}$ contribute $2^{(1-a-b)n}-1$ to the cut if $\Phi\vDash C_i$, and $2^{(1-a-b)n}$ otherwise. Since there are at most $m-p$ clauses that are not satisfied by $\Phi$, summing over all $i\in [m]$ would yield a total of at most $p(2^{(1-a-b)n}-1)+(m-p)(2^{(1-a-b)n})=2^{(1-a-b)n}m-p$ cut edges for vertices with origin in $A$.

For every $\beta\neq\beta_{\Phi}$, all nodes in $V_2$ that are associated with $\beta$, $v_B$, and $\gamma_{\Phi}$, are in $T$ and hence will not contribute anything to the cut. However, the node ${\beta_{\Phi}}'$ is always in $S$, with $v_B$ its sole target, and hence the edge $({\beta_{\Phi}}',v_B)$ is in the cut and ${\beta_{\Phi}}'$ contributes an additional amount of $p-1$, to a current total of at most $2^{(1-a-b)n}m-p+(p-1)=2^{(1-a-b)n}m-1$. In addition, ${\beta_{\Phi}}'$ is the only target of ${(\beta_{\Phi})}_i^c$, and thus ${(\beta_{\Phi})}_i^c$ will not contribute to the cut.

We will show that the rest of the nodes, i.e., nodes in $V_2$ that are of the forms $\beta_{\Phi}^l$ and $\beta_{\Phi}^l$, and the nodes in $B$, contribute nothing to the cut. For every $i\in [m]$, ${(\beta_{\Phi})}_i^l\in S$ if and only if either $\beta_{\Phi}\vDash C_i$ or $\gamma_{\Phi}\vDash C_i$, so we assume that. It always happens that ${(\beta_{\Phi})}_i^c\in S$, and ${(\beta_{\Phi})}_i^r\in T$ if and only if $\gamma_{\Phi}\nvDash C_i$, but in such case, by our assumption it must be that $\beta_{\Phi}\vDash C_i$, which implies that the edge $({(\beta_{\Phi})}_i^l,{(\beta_{\Phi})}_i^r)$ is not in the graph, thus the total contribution of ${(\beta_{\Phi})}_i^l$ is zero. Continuing to nodes of the forms ${(\beta_{\Phi})}_i^r$ and $C_i$, it is easy to verify that the following four statements are either all true or all false: ${(\beta_{\Phi})}_i^r\in S$, $\gamma_{\Phi}\vDash C_i$, $C_i\in S$, and the edge $(C_i,\gamma_{\Phi})$ is not in the graph. In the case where they all false, in particular $C_i$ and ${(\beta_{\Phi})}_i^r$ are in $T$ and it is clear that they do not contribute anything, so we will focus on the remaining case. Since $C_i$ is in $S$ and is the only target of ${(\beta_{\Phi})}_i^r$, ${(\beta_{\Phi})}_i^r$ will not increase the cut capacity. In addition, since the edge $(C_i,\gamma_{\Phi})$ is not in the graph and also all nodes $V_3\setminus \{\gamma_{\Phi}\}$ are in $S$, $C_i$ does not increase the capacity of the cut either. Altogether we have bounded the total capacity of the cut by $2^{(1-a-b)n}m-1$, finishing the proof of Claim~\ref{CutLemma}.
\end{proof}

Proceeding with the proof of Lemma~\ref{ProofsCapLemma}, we now focus on the second direction. Assume that every assignment to $F$ satisfies at most $p-1$ clauses. We remind that we need to prove that the maximum flow from every $\alpha$ to every $\gamma$ is at least $2^{(1-a-b)n}m$, and to do this we first fix $\alpha$ and $\gamma$. By the assumption, for every $\beta$ there exist $m-(p-1) = m-p+1$ $i$'s, such that $\alpha$, $\beta$, and $\gamma$ do not satisfy $C_i$, and we denote a set with this amount of such $i$'s by $I_{\beta}$. Each of these $i$'s induces a blue path $(\alpha\rightarrow C_i^{\nvDash}\rightarrow \beta_i^l\rightarrow \beta_i^r\rightarrow C_i\rightarrow \gamma)$ from $\alpha$ to $\gamma$, and so we pass a unit of flow through every one of them according to $I_{\beta}$, and for all $\beta$, in what we call the first phase. We note that so far, the flow sums up to $2^{(1-a-b)n}(m-p+1)$, and so we carry on with shipping the second phase of flow through paths that are not entirely blue.

We claim that for every $\beta$, we can pass an additional amount of $m-(m-p+1) = p-1$ units through $\beta'$, which would add up to a total flow of $2^{(1-a-b)n}(m-p+1)+2^{(1-a-b)n}(p-1)=2^{(1-a-b)n}m$, concluding the proof. Indeed, for every $\beta$, we ship flow in the following way. For every $i\in [m]\setminus I_{\beta}$, if $\alpha\nvDash C_i$ then send a unit through $(\alpha\rightarrow C_i^{\nvDash}\rightarrow \beta_{i}^l\rightarrow \beta_i^c\rightarrow \beta'\rightarrow v_B\rightarrow \gamma)$, and otherwise send a unit through $(\alpha\rightarrow C_i^{\vDash}\rightarrow \beta_i^c \rightarrow \beta'\rightarrow v_B\rightarrow \gamma)$. 

Since we defined the flow in paths, we only need to show that the capacity constraints are satisfied, starting with edges of color blue. Edges of the forms $(\beta_i^l,\beta_i^r)$, $(\beta_i^r,C_i)$, and $(C_i,\gamma)$ are only used in the first phase, where the flow in the first two is uniquely determined by $\beta$ and $i\in I_{\beta}$, and so at most $1$ unit of flow is passed through them, and the flow in the latter kind is determined by $i\in I_{\beta}$, and the same $i\in I_{\beta}$ can have at most $\card{\{\beta_i^r\}_{\beta}}=2^{(1-a-b)n}$ units of flow passing in $(C_i, \gamma)$, and so the flow in it is also bounded. The flow in edges of the form $(C_i^{\nvDash}, \beta_i^l)$ in the first phase is uniquely determined by $\beta$ and $i\in I_{\beta}$, and in the second phase uniquely according to $\beta$ and $i\in [m]\setminus I_{\beta}$, and so will not be used twice, and the flow in edges of the form $(\alpha, C_i^{\nvDash})$ is determined in the first phase by $i\in I_{\beta}$ and in the second phase by $i\in [m]\setminus I_{\beta}$, and so will be used at most $\sum_{\beta} \card{I_{\beta}\cap \{i\}}+\sum_{\beta}\card{([m]\setminus I_{\beta})\cap \{i\} }\leq 2^{(1-a-b)n}$ times.

We now proceed to prove that red edges too do not have more flow than their capacity, and for this we only need to consider the second phase. Edges of the forms $(C_i^{\vDash}, \beta_i^c)$, $(\beta_i^l,\beta_i^c)$, and $(\beta_i^c,\beta')$ have flow that is uniquely determined by $\beta$ and $i\in [m]\setminus I_{\beta}$ and so are not used more than once, edges of the form $(\beta',v_B)$ have flow that is determined by $\beta$ and thus have flow $\card{\{\beta_i^c\}_{i\in [m]\setminus I_{\beta}}}=\card{[m]\setminus I_{\beta}} = p-1$, and edges of the form $(v_B,\gamma)$ have flow of size $(p-1)\card{\{\beta'\}_{\beta}}2^{(1-a-b)n}=(p-1)2^{(1-a-b)n}$, and hence are properly bounded. Finally, edges of the form $(\alpha,C_i^{\vDash})$ have flow that is determined by $i\in [m]\setminus I_{\beta}$ and so are used at most $\card{\{\beta_i^c\}_{\beta}}=2^{(1-a-b)n}$ times. Altogether, we have bounded the total flow in all the edges that were used in both phases, and so the capacity requirements follow, which completes the proof of the second direction and of Lemma~\ref{ProofsCapLemma}.
\end{proof}

\begin{proof}[Proof of Theorem~\ref{ST}]
We apply Lemma~\ref{ProofsCapLemma} in the following way. For every setting of $a,b\in [0,1/2]$ such that $1-a-b\geq\max(a,b)$ we get graphs $G=(V,E,w)$ with $\card{V}=\Theta(2^{(1-a-b)n}m)$ and $\card{E}=O(2^{(1-a-b)n}m)=O(\card{V})$. Hence, in order to get any $c_1, c_2\in [0,1]$, we can pick $a=c_1/(1+c_1+c_2)$ and $b=c_2/(1+c_1+c_2)$, which clearly satisfy the required conditions. Now, observe that $c_1=a/(1-a-b)$ and $c_2=b/(1-a-b)$, thus $\card{S}=(\card{V}/m)^{c_1}$ and $\card{T}=(\card{V}/m)^{c_2}$, and our claimed lower bound and Theorem~\ref{ST} follow.
%
%

\end{proof}
\section{Generalization to Bounded Cuts}\label{GTSC}
Our lower bounds extend to the version where we only care about vertex-pairs with maximum flow bounded by a given $k$, 
which we refer to as \ProblemName{kPMF}.

\begin{definition}(\kPMF)
Given a directed edge-capacitated graph $G=(V,E,w)$ and an integer $k$, for every pair of nodes $u,v\in V$ where the maximum flow that can be shipped in $G$ from $u$ to $v$ is of size at most $k$, output this pair and its maximum flow value.
\end{definition}

\begin{theorem}[Generalization of Theorem~\ref{UncapAllPairs}]
\label{Uncapacitated_kPMF}
If for some fixed constants $\varepsilon>0$ and $c\in [0, 1]$, \kPMF in uncapacitated graphs with $n$ nodes, $k=\tO(n^c)$, and $m=O(kn)$ edges can be solved in time $O((n^2 k)^{1-\varepsilon})$, then for some $\delta(\varepsilon)>0$, MAX-CNF-SAT on $n'$ variables and $O(n')$ clauses can be solved in time $O(2^{(1-\delta)n'})$, and in particular SETH is false.
\end{theorem}

\begin{proof}
We apply Lemma~\ref{ProofsUncapLemma} as follows. For every setting of $a=b\in[1/3,1/2]$ we get graphs $G=(V,E,w)$ with $\card{V}=\Theta(2^{an})$ ($\card{V}=\Theta(2^{an}m)$ if $a=1/3$), and $\card{E}=2^{an}\cdot 2^{(1-2a)n}m=\Theta(2^{(1-a)n}m)$. The main idea is that the middle layer bound the flow from every $\alpha$ to every $\gamma$, which are the only pairs that we need to find the maximum flow for. To be more precise, for every $\alpha'$ and $\gamma'$ we show a cut of capacity $k=O(2^{(1-2a)n}m)$ separating them, by considering
$$
S=\{\alpha'\}\cup
\{{\beta}_i^l : i\in [m], \forall \beta \}.
$$
Clearly, the only outgoing edges from $S$ are from $\alpha'$ and from vertices of the form ${\beta}_i^l$. $\alpha'$ has an outgoing degree at most $O(2^{(1-2a)n}m)$, and for each $\beta$ and $i\in [m]$, vertices of the form ${\beta}_i^l$ have a total outgoing degree at most $2$. Hence, the total capacity of the cut is bounded by $k=O(2^{(1-2a)n}m)$.
The claimed range of $k$ is attained 
because setting $a=1/2$ yields $k=O(m)=O(\log\card{V})\leq O(n^{c})$, 
and letting $a$ approach $1/3$ yields $k$ tending to $O(2^{n/3}m)=O(\card{V})$.
Note that $\card{E}=O(\card{V}k)$, and $\card{V}^2k=O((2^{an})^2\cdot 2^{(1-2a)n}m)=O(2^nm)$, and finally in order to get any $c\in [0,1]$ we can pick $a(=b)$ such that additionally $c=1/a-2$, and Theorem~\ref{Uncapacitated_kPMF} holds.
\end{proof}

\begin{theorem}[Generalization of Theorem~\ref{AllPairs}] 
\label{Capacitated_kPMF}
If for some fixed constants $\varepsilon>0$ and $c\in [0,1]$, \kPMF in graphs with $n$ nodes, $k=\tO(n^c)$, $m=O(n)$ edges, and capacities in $[n]$ can be solved in time $O((n^2 k)^{1-\varepsilon})$, then for some $\delta(\varepsilon)>0$, MAX-CNF-SAT on $n'$ variables and $O(n')$ clauses can be solved in time $O(2^{(1-\delta)n'})$, and in particular SETH is false.
\end{theorem}
\begin{proof}

We apply Lemma~\ref{ProofsCapLemma} in a similar fashion to the application of Lemma~\ref{ProofsUncapLemma} in the proof of Theorem~\ref{Uncapacitated_kPMF}, where the choices of $a$ and $b$ are done in exactly the same way as before, also allowing again a free choice of $c\in[0,1]$. However, now $\card{E}=O(\card{V})$, and we choose the cut as follows. For every $\alpha'$ and $\gamma'$ we show a cut of capacity $k=O(2^{(1-2a)n}m)$ separating them, by considering

$$
S=\{\alpha'\}\cup
\{C_i^{\vDash},C_i^{\nvDash} : i\in [m] \}\cup 
\{{\beta}_i^l, {\beta}_i^c : i\in [m], \forall \beta \}.
$$

Clearly, the only outgoing edges from $S$ are of capacity $1$, from $\alpha'$ and from vertices of the forms ${\beta}_i^l$ and ${\beta}_i^c$. For each $\beta$ and $i\in [m]$, these vertices have a total of at most $2$ edges going out to the rest of the graph. Hence, the size of the cut is bounded by $k=O(2^{(1-2a)n}m)$, and the range of $k$ is similar to the proof of Theorem~\ref{Uncapacitated_kPMF}, and so Theorem~\ref{Capacitated_kPMF} holds.


%
\end{proof}

Known algorithms solve \kPMF in \emph{directed} graphs
in time $\tO(n^2m\cdot \min(k, \sqrt{n}))$, which is bigger than the lower bound in Theorem~\ref{Capacitated_kPMF} by a factor that is roughly between $\sqrt{n}$ and $n$ for sparse graphs, 
leaving a gap that is not too big even for relatively small values of $k$. 
This running time can be achieved by $O(n^2)$ computations of either the aforementioned $O(mk)$ time algorithm of~\cite{FF56} (actually, a slightly modified version that halts when the total flow exceeds $k$), or the $\tO(m\sqrt{n})$ time algorithm of~\cite{lee2014path}.

It is interesting to note that in graphs that are \textit{undirected} and uncapacitated, an algorithm for \kPMF with running time $O(mk+n^2)$ was shown in~\cite{BHKP07}. This shows a separation between the directed and the undirected cases also for uncapacitated graphs, 
roughly by a factor $\Omega(n^2k/(mk+n^2))=\Omega(\min(k, n/k))$, since our relevant conditional lower bound is proved for $m=O(kn)$.
Their algorithm actually builds in time $O(mk)$ a partial Gomory-Hu tree that succinctly represents the values required by \kPMF, 
and then it is easy to extract all the relevant values in time $O(n^2)$, 
as required by our definition of \kPMF. 
For instance, when $k=O(\sqrt{n})$ and $m=O(n^{3/2})$ their upper bound for the undirected and uncapacitated case is $O(n^2)$, while our lower bound for the directed case is $n^{2.5-o(1)}$.

\section{\ProblemName{Global Max-Flow}}\label{GMF}
\begin{proof}[Proof of Theorem~\ref{MLEC}]
Given a CNF-formula $F$ on $n$ variables and $m$ clauses $\{C_i\}_{i\in [m]}$ as input for MAX-CNF-SAT, we split the variables into two sets $U_1$ and $U_2$ of size $n/2$ each and enumerate all $2^{n/2}$ partial assignments (with respect to $F$) to each of them, when the objective is to find a pair $(\alpha, \beta)$ of assignments to $U_1$ and $U_2$ that satisfy the maximal number of clauses. We construct a graph $G=(V,E)$ such that $V=L\cup R\cup C$ as follows. $L$ contains a node $\alpha$ for every assignment $\alpha$ to $U_1$, $R$ contains a node $\beta$ for every assignment $\beta$ to $U_2$, and $C$ contains three nodes $c_{\vDash,\vDash}$, $c_{\vDash,\nvDash}$, and $c_{\nvDash,\vDash}$ for every clause $C_i$. We use the notation $\alpha$ for nodes in $L$ and assignments to $U_1$, $\beta$ for nodes in $R$ and assignments to $U_2$. However, it will be clear from the context. For every assignment $\alpha$ to $U_1$ and clause $C_i$, we add an edge from $\alpha$ to $c_{\vDash,\vDash}$ and $c_{\vDash,\nvDash}$ if $\alpha \vDash C_i$, and an edge from $\alpha$ to $c_{\nvDash,\vDash}$ otherwise. Similarly, for every assignment $\beta$ to $U_2$ and clause $C_i$, we add an edge from $\beta$ to $c_{\vDash,\vDash}$ and $c_{\nvDash,\vDash}$ if $\beta\vDash C_i$, and an edge from $\beta$ to $c_{\vDash,\nvDash}$ otherwise. 
This graph has $N = 2^{n/2} + 2^{n/2} + 3m = O(2^{n/2})$ nodes and at most $N\cdot 2m + N\cdot 2m = \tO(N)$ edges.
%
For every pair of assignments $\alpha$ and $\beta$ and clause $C_i$ there is exactly one path (of length $2$) from $\alpha$ to $\beta$ through nodes associated with $C_i$ if and only if $\alpha\vDash C_i$ or $\beta\vDash C_i$, and no paths through them otherwise. Hence, the number of edge disjoint paths from $\alpha$ to $\beta$ is exactly the number of clauses that are satisfied by both of the assignments $\alpha$ and $\beta$, and so an algorithm for \MLEC with running time $\tO(n^{2-\varepsilon})$ implies an algorithm for MAX-CNF-SAT with running time $\tO((2^{n/2})^{2-\varepsilon})=\tO(2^{(1-\varepsilon/2)n})$, completing the proof for $\delta(\varepsilon)=\varepsilon/2$.
\end{proof}

\paragraph*{Acknowledgements}
We thank Rajesh Chitnis and Bundit Laekhanukit for some useful conversations, and for their part in achieving the result on \GMF.

\bibliographystyle{alphaurlinit}

\bibliography{robi,drafts}

\newcommand{\etalchar}[1]{$^{#1}$}
\begin{thebibliography}{BENW16}

\bibitem[ABHS17]{abboud2017seth}
A.~Abboud, K.~Bringmann, D.~Hermelin, and D.~Shabtay.
\newblock {SETH}-based lower bounds for subset sum and bicriteria path.
\newblock {\em CoRR}, 2017.
\newblock Available from: \url{http://arxiv.org/abs/1704.04546}.

\bibitem[ACZ98]{ArikatiCZ95}
S.~R. Arikati, S.~Chaudhuri, and C.~D. Zaroliagis.
\newblock All-pairs min-cut in sparse networks.
\newblock {\em J. Algorithms}, 29(1):82--110, 1998.
\newblock \href {http://dx.doi.org/10.1006/jagm.1998.0961}
  {\path{doi:10.1006/jagm.1998.0961}}.

\bibitem[AHK12]{arora2012multiplicative}
S.~Arora, E.~Hazan, and S.~Kale.
\newblock The multiplicative weights update method: a meta-algorithm and
  applications.
\newblock {\em Theory of Computing}, 8(1):121--164, 2012.
\newblock \href {http://dx.doi.org/10.4086/toc.2012.v008a006}
  {\path{doi:10.4086/toc.2012.v008a006}}.

\bibitem[AMO93]{AMJ93}
R.~K. Ahuja, T.~L. Magnanti, and J.~B. Orlin.
\newblock {\em Network flows - theory, algorithms and applications}.
\newblock Prentice Hall, 1993.

\bibitem[AVY15]{AbboudWY15}
A.~Abboud, V.~{Vassilevska-Williams}, and H.~Yu.
\newblock Matching triangles and basing hardness on an extremely popular
  conjecture.
\newblock In {\em Proceedings of the Forty-seventh Annual ACM Symposium on
  Theory of Computing}, STOC '15, pages 41--50. ACM, 2015.
\newblock \href {http://dx.doi.org/10.1145/2746539.2746594}
  {\path{doi:10.1145/2746539.2746594}}.

\bibitem[AWY15]{Abboud15}
A.~Abboud, R.~Williams, and H.~Yu.
\newblock More applications of the polynomial method to algorithm design.
\newblock In {\em Proceedings of the Twenty-sixth Annual ACM-SIAM Symposium on
  Discrete Algorithms}, SODA '15, pages 218--230, 2015.
\newblock \href {http://dx.doi.org/10.1145/2722129.2722146}
  {\path{doi:10.1145/2722129.2722146}}.

\bibitem[BENW16]{borradaile2014all}
G.~Borradaile, D.~Eppstein, A.~Nayyeri, and C.~{Wulff-Nilsen}.
\newblock All-pairs minimum cuts in near-linear time for surface-embedded
  graphs.
\newblock In {\em 32nd International Symposium on Computational Geometry (SoCG
  2016)}, volume~51 of {\em Leibniz International Proceedings in Informatics
  (LIPIcs)}, pages 22:1--22:16. Schloss Dagstuhl--Leibniz-Zentrum fuer
  Informatik, 2016.
\newblock \href {http://dx.doi.org/10.4230/LIPIcs.SoCG.2016.22}
  {\path{doi:10.4230/LIPIcs.SoCG.2016.22}}.

\bibitem[BHKP07]{BHKP07}
A.~Bhalgat, R.~Hariharan, T.~Kavitha, and D.~Panigrahi.
\newblock An {$\tilde O(mn)$} {G}omory-{H}u tree construction algorithm for
  unweighted graphs.
\newblock In {\em 39th Annual ACM Symposium on Theory of Computing}, STOC'07,
  pages 605--614. ACM, 2007.
\newblock \href {http://dx.doi.org/10.1145/1250790.1250879}
  {\path{doi:10.1145/1250790.1250879}}.

\bibitem[BJS10]{bazaraa2011linear}
M.~S. Bazaraa, J.~J. Jarvis, and H.~D. Sherali.
\newblock {\em Linear Programming and Network Flows}.
\newblock John Wiley \& Sons, Inc., fourth edition, 2010.

\bibitem[BK09]{borradaile2009n}
G.~Borradaile and P.~Klein.
\newblock An $\tilde{O}(n\log n)$ algorithm for maximum $st$-flow in a directed
  planar graph.
\newblock {\em J. ACM}, 56(2):9:1--9:30, 2009.
\newblock \href {http://dx.doi.org/10.1145/1502793.1502798}
  {\path{doi:10.1145/1502793.1502798}}.

\bibitem[CGI{\etalchar{+}}16]{carmosino2016nondeterministic}
M.~L. Carmosino, J.~Gao, R.~Impagliazzo, I.~Mihajlin, R.~Paturi, and
  S.~Schneider.
\newblock Nondeterministic extensions of the strong exponential time hypothesis
  and consequences for non-reducibility.
\newblock In {\em Proceedings of the 2016 ACM Conference on Innovations in
  Theoretical Computer Science}, ITCS '16, pages 261--270. ACM, 2016.
\newblock \href {http://dx.doi.org/10.1145/2840728.2840746}
  {\path{doi:10.1145/2840728.2840746}}.

\bibitem[CLL11]{cheung2013graph}
H.~Y. Cheung, L.~C. Lau, and K.~M. Leung.
\newblock Graph connectivities, network coding, and expander graphs.
\newblock In {\em Proceedings of the 2011 IEEE 52nd Annual Symposium on
  Foundations of Computer Science}, FOCS '11, pages 190--199. IEEE Computer
  Society, 2011.
\newblock \href {http://dx.doi.org/10.1109/FOCS.2011.55}
  {\path{doi:10.1109/FOCS.2011.55}}.

\bibitem[CW16]{Chan16}
T.~M. Chan and R.~Williams.
\newblock Deterministic apsp, orthogonal vectors, and more: Quickly
  derandomizing razborov-smolensky.
\newblock In {\em Proceedings of the Twenty-seventh Annual ACM-SIAM Symposium
  on Discrete Algorithms}, SODA '16, pages 1246--1255, 2016.
\newblock \href {http://dx.doi.org/10.1145/2884435.2884522}
  {\path{doi:10.1145/2884435.2884522}}.

\bibitem[FF56]{FF56}
L.~R. Ford, Jr. and D.~R. Fulkerson.
\newblock Maximal flow through a network.
\newblock {\em Canadian Journal of Mathematics}, 8:399--404, 1956.
\newblock \href {http://dx.doi.org/10.4153/CJM-1956-045-5}
  {\path{doi:10.4153/CJM-1956-045-5}}.

\bibitem[Fre95]{Frederickson95}
G.~N. Frederickson.
\newblock Using cellular graph embeddings in solving all pairs shortest paths
  problems.
\newblock {\em J. Algorithms}, 19(1):45--85, 1995.
\newblock \href {http://dx.doi.org/10.1006/jagm.1995.1027}
  {\path{doi:10.1006/jagm.1995.1027}}.

\bibitem[GH61]{GH61}
R.~E. Gomory and T.~C. Hu.
\newblock Multi-terminal network flows.
\newblock {\em Journal of the Society for Industrial and Applied Mathematics},
  9:551--570, 1961.
\newblock \href {http://dx.doi.org/10.1137/0109047}
  {\path{doi:10.1137/0109047}}.

\bibitem[Gus90]{Gusfield90}
D.~Gusfield.
\newblock Very simple methods for all pairs network flow analysis.
\newblock {\em SIAM J. Comput.}, 19(1):143--155, 1990.
\newblock \href {http://dx.doi.org/10.1137/0219009}
  {\path{doi:10.1137/0219009}}.

\bibitem[HO94]{hao1994faster}
J.~Hao and J.~B. Orlin.
\newblock A faster algorithm for finding the minimum cut in a directed graph.
\newblock {\em J. Algorithms}, 17(3):424--446, 1994.
\newblock \href {http://dx.doi.org/10.1006/jagm.1994.1043}
  {\path{doi:10.1006/jagm.1994.1043}}.

\bibitem[IP01]{ImpaSETH}
R.~Impagliazzo and R.~Paturi.
\newblock On the complexity of k-{SAT}.
\newblock {\em Journal of Computer and System Sciences}, 62(2):367--375, 2001.
\newblock \href {http://dx.doi.org/10.1006/jcss.2000.1727}
  {\path{doi:10.1006/jcss.2000.1727}}.

\bibitem[IPZ01]{Impa01spar}
R.~Impagliazzo, R.~Paturi, and F.~Zane.
\newblock Which problems have strongly exponential complexity?
\newblock {\em Journal of Computer and System Sciences}, 63(4):512--530, 2001.
\newblock \href {http://dx.doi.org/10.1006/jcss.2001.1774}
  {\path{doi:10.1006/jcss.2001.1774}}.

\bibitem[KL02]{KM02}
D.~R. Karger and M.~S. Levine.
\newblock Random sampling in residual graphs.
\newblock In {\em Proceedings of the Thiry-fourth Annual ACM Symposium on
  Theory of Computing}, STOC '02, pages 63--66, New York, NY, USA, 2002. ACM.
\newblock \href {http://dx.doi.org/10.1145/509907.509918}
  {\path{doi:10.1145/509907.509918}}.

\bibitem[LNSW12]{lacki2012single}
J.~Lacki, Y.~Nussbaum, P.~Sankowski, and C.~{Wulff-Nilsen}.
\newblock Single source -- all sinks max flows in planar digraphs.
\newblock In {\em Proceedings of the 53rd IEEE Annual Symposium on Foundations
  of Computer Science}, FOCS '12, pages 599--608. IEEE Computer Society, 2012.
\newblock \href {http://dx.doi.org/10.1109/FOCS.2012.66}
  {\path{doi:10.1109/FOCS.2012.66}}.

\bibitem[LS14]{lee2014path}
Y.~T. Lee and A.~Sidford.
\newblock Path finding methods for linear programming: Solving linear programs
  in \~{O}($\sqrt{\mathrm{rank}}$) iterations and faster algorithms for maximum
  flow.
\newblock In {\em Proceedings of the 2014 IEEE 55th Annual Symposium on
  Foundations of Computer Science}, FOCS '14, pages 424--433. IEEE Computer
  Society, 2014.
\newblock \href {http://dx.doi.org/10.1109/FOCS.2014.52}
  {\path{doi:10.1109/FOCS.2014.52}}.

\bibitem[M{\k{a}}d16]{madry2016computing}
A.~M{\k{a}}dry.
\newblock Computing maximum flow with augmenting electrical flows.
\newblock In {\em Proceedings of the 57th IEEE Annual Symposium on Foundations
  of Computer Science}, FOCS '16, pages 593--602. IEEE Computer Society, 2016.
\newblock \href {http://dx.doi.org/10.1109/FOCS.2016.70}
  {\path{doi:10.1109/FOCS.2016.70}}.

\bibitem[Sch02]{schrijver2002history}
A.~Schrijver.
\newblock On the history of the transportation and maximum flow problems.
\newblock {\em Math. Program.}, 91(3):437--445, 2002.
\newblock \href {http://dx.doi.org/10.1007/s101070100259}
  {\path{doi:10.1007/s101070100259}}.

\bibitem[{Vas}15]{VWSurvey}
V.~{Vassilevska-Williams}.
\newblock {Hardness of Easy Problems: Basing Hardness on Popular Conjectures
  such as the Strong Exponential Time Hypothesis (Invited Talk)}.
\newblock In {\em 10th International Symposium on Parameterized and Exact
  Computation (IPEC 2015)}, volume~43 of {\em Leibniz International Proceedings
  in Informatics (LIPIcs)}, pages 17--29. Schloss Dagstuhl--Leibniz-Zentrum
  fuer Informatik, 2015.
\newblock \href {http://dx.doi.org/10.4230/LIPIcs.IPEC.2015.17}
  {\path{doi:10.4230/LIPIcs.IPEC.2015.17}}.

\bibitem[{Vas}18]{Vsurvey18}
V.~{Vassilevska-Williams}.
\newblock On some fine-grained questions in algorithms and complexity.
\newblock In {\em Proceedings of ICM}, 2018.
\newblock To Appear.
\newblock Available from: \url{http://people.csail.mit.edu/virgi/eccentri.pdf}.

\bibitem[Wil05]{WilliamsOV}
R.~Williams.
\newblock A new algorithm for optimal 2-constraint satisfaction and its
  implications.
\newblock {\em Theoretical Computer Science}, 348(2):357--365, 2005.
\newblock \href {http://dx.doi.org/10.1016/j.tcs.2005.09.023}
  {\path{doi:10.1016/j.tcs.2005.09.023}}.

\end{thebibliography}

\end{document}